\documentclass[a4paper,11pt]{amsart}
\usepackage{enumerate}

\theoremstyle{plain}
\newtheorem{theorem}{Theorem}[section]
\newtheorem{lemma}[theorem]{Lemma}
\newtheorem{corollary}[theorem]{Corollary}

\theoremstyle{definition}

\theoremstyle{remark}
\newtheorem{remark}{Remark}

\newcommand{\C}{\mathbb{C}}

\begin{document}

\title[periodic solutions]
      {In search of periodic solutions for a reduction of the Benney chain}

\date{20 April 2017}
\author{Misha Bialy and Andrey E. Mironov}
\address{M. Bialy, School of Mathematical Sciences,
Raymond and Beverly Sackler Faculty of Exact Sciences, Tel Aviv University,
Israel} \email{bialy@post.tau.ac.il}
\address{A.E. Mironov,  Sobolev Institute of Mathematics,
Academician Koptyug avenue 4, 630090,
Novosibirsk, Russia and Department of Mathematics and Mechanics,
Novosibirsk State University,
Pirogov street 2, 630090 Novosibirsk, Russia}
\email{mironov@math.nsc.ru}
\thanks{M.B. was supported in part by ISF grant 162/15 and A.M. was
supported by RSF (grant
14-11-00441). It is our pleasure to thank these funds for the
support}


\begin{abstract} We search for smooth periodic solutions for the system
of quasi-linear PDEs known as the Lax dispersionless reduction of the Benney
moments chain. It is naturally related to the existence of a
polynomial in momenta integral for a Classical Hamiltonian system
with 1,5 degrees of freedom. For the solution in question it is not
known a priori if the system is elliptic or hyperbolic or of mixed
type. We consider two possible regimes for the solution. The first
is the case of only one real eigenvalue, where we can completely
classify the solutions. The second case of strict Hyperbolicity is
really a challenge.  We find a remarkable 2 by 2 reduction which is
strictly Hyperbolic but violates the condition of genuine
non-linearity.
\end{abstract}

\maketitle

\section{Motivation and the results}

\label{sec:intro} The famous equations of Benney moments are an
infinite system of PDEs on the functions $A^k(t,x),\  k=0,1,2...$
$$
A^k_t+A^{k+1}_x+kA^{k-1} A^0_x=0, \ k=0,1,2, . . . .
$$
(\cite {benney}, \cite {kupershmidt-manin}). It admits many
reductions where infinitely many functions $A^n$ become the
functions on finitely many field variables $U=(u_1,..,u_n)$ (see
\cite{lorenzoni}, \cite{gt}, \cite{gk},
\cite{pavlov}). In this paper we deal with one of the reductions
called the dispersionless Lax reduction, where $A_i$ are expressed via
$U=(u_1,..,u_n)$ by the formula:

$$p+\sum_{k=0}^{\infty}\frac{A^k}{p^{k+1}}=
\left(p^{n+1}+(n+1)u_1p^{n-1}+...+(n+1)u_n\right)^{\frac{1}{n+1}}.$$
It can be checked that the equations on $(u_1,..,u_n)$
can be written explicitly as a quasi-linear system of the form:
\begin{equation}
  \label{system}
  U_t+A(U)U_x=0 , \
  A(U)=-
  \begin{pmatrix}
  0 & -1 & 0 &  \cdots & 0 & 0 \\[1mm]
  {(n-1)}u_1 & 0 & -1 & 0 & \cdots & 0 \\[1mm]
  \vdots & \vdots & \vdots & \vdots & \vdots & \vdots \\
  2u_{n-2} & 0 & 0 & \cdots & 0 & -1\\
  u_{n-1} & 0 & 0 & \cdots & 0 & 0
  \end{pmatrix},
\end{equation}
for unknown functions $U=(u_1 ,.., u_n)^t$.

This is a remarkable Hamiltonian system of the hydrodynamic type
related to the $A_n$ singularity (see \cite{dubrovin-novikov},
\cite{dubrovin}, \cite{strachan}, \cite{lorenzoni}).
Though,
 local/formal solutions for this system can be studied by the so called
 generalized hodograph method (\cite{tsarev}, \cite{gk}),
 very little is known on the existence of smooth solutions for
 (\ref{system}) globally.
 General belief is that they are very rare.
 In \cite{bialy1}, \cite{bialy2}, \cite{bialy3} the global analysis of periodic
 smooth solution was started for (\ref{system}) and also in (\cite{BM1}, \cite{BM2}) for
 for another Semi-hamiltonian system corresponding to geodesic flows. Here we continue in this direction
 further.

 It was observed in \cite{kozlov} and later in \cite{bialy1}, \cite{bialy2}, \cite{bialy3} that the question of existence
 of smooth periodic solution for (\ref{system}) is ultimately related to the search of
 polynomial integrals for a Hamiltonian system
  as we now turn to explain.

Let $ H=p^2/2+u(t,x) $ be a Hamiltonian of a $1,5$-degrees of
freedom system with the potential $u$ which is assumed throughout this
paper to be $C^2-$smooth periodic function in both variables. The only known examples of   integrable Hamiltonian
of this form with periodic potential functions $u$, are those having the potential $u$ of the form of traveling waves:
$u=u(mx+nt)$.

More precisely we want to find all those potential functions
$u(t,x)$ for which there exists an additional function $F(p,x,t)$
invariant under the Hamiltonian flow (such an $F$ is called the
first integral of motion). Let us stick to the case where $F$ is a
polynomial in the variable $p$ of a given degree, say $(n+1)$,
having all the coefficients $C^2-$smooth, periodic in $x$ and $t$.
Write
$$ F(p,x,t)=u_{-1}p^{n+1}+u_{0}p^n+u_1p^{n-1}+\cdots +u_n,$$
and substitute to the equation of conservation of $F$
\begin{equation}
 \label{eq:conservation}
 F_t+pF_x-u_xF_p=0.
\end{equation}
Equating to zero the coefficients of various powers of $p$, one
easily obtains the following information. The coefficient $u_{-1}$
must be a constant, which can be normalized to be ${1\over n+1}$.
Also $u_0$ must be a constant, which we shall assume to be zero
(this can be achieved by a linear change of coordinates on the
configuration space $\mathbb{T}^2$). Moreover the coefficient $u_1$
satisfies $(u_1)_x=(u)_x$. Therefore, $u_1$ and $u$ will be assumed
to be equal (the addition of any function of $t$ to the potential
$u$ does not change the Hamiltonian equations). Moreover, the
functions $U=(u_1,...,u_n)$ satisfy precisely the system
(\ref{system}).

It is remarkable that the system (\ref{system}) is Hamiltonian and in particular
belongs to the class of Semi-hamiltonian or Rich. This means it can
be written in terms of Riemann invariants and has infinitely many
conservation laws (\cite{tsarev}, \cite{Sev}, \cite{Serre}).

Let us emphasize that the study of smooth solutions of the system is
a very challenging one, for the two reasons: the first reason is
that the system is of mixed type a-priori, since it depends on the
solution in question. Moreover, in the strictly hyperbolic region it fails to be
genuinely non-linear, because the sign of the non-linearity for some
eigenvalues can change.

It is important to observe that the characteristic polynomial
 of $A(U)$ has very clear geometric
meaning \cite{bialy1}, \cite{bialy2}, \cite{bialy3}: it coincides with
derivative $F_p$ of $F$, thus giving the information on the phase
portrait of the system with the Hamiltonian $H$. Namely, the graph of an
eigenvalue $p=\lambda(t,x)$ has the property that invariant torii of
the hamiltonian flow have vertical tangents at the points on the
graph.

 The purpose of this paper is to study two opposite possible
 regimes for the system (\ref{system}).
 In the first part we consider the case when only one
 eigenvalue of the matrix $A(U)$ is real and the rest are
 complex conjugate pairs (not real).
 In the second part we deal with the strictly hyperbolic regime, when all eigenvalues of the matrix $A(U)$ are real and distinct.
 
  We turn now to  formulation of our result for the first case. Let us remark first that the assumption of one real eigenvalue for the system (\ref{system}) is very
 natural in view of phase portrait of autonomous
 Hamiltonian with 1-degree of freedom, where only one chain of separatix islands is present.

 Notice that complex eigenvalues are allowed to collide
 for some $(t,x)$, however we shall assume, in this case, that
 the characteristic polynomial can be factorized in a continuous way,
 see remarks below. Our main result in this case is that the only
 solutions in this regime are autonomous ones. This is formulated in
 the following theorem:

\begin{theorem}
 \label{main theorem} {\it Let $n=2l+1$ be odd. Assume that periodic solution $U(t,x)$ be such that
 the matrix $A(U)$ has one
 real and $l$ complex conjugate pairs of eigenvalues for every $(t,x)$.
 We assume that
 the characteristic polynomial of $A(U)$ can be continuously
 factorized:
 $$
F_p=(p-\mu)\left((p-\lambda_1)(p-\bar{\lambda}_1)...(p-\lambda_l)
(p-\bar{\lambda}_l)\right),
 $$
where $\mu(t,x), \lambda_i(t,x)$ are continuous functions and $\mu$
is real and $\lambda_i, \bar{\lambda}_i$ are complex conjugate
pairs.
 Then the
solution $U$ of
 quasi-linear system (\ref {system}) is the traveling wave
 solution, where the components $(u_1,...,u_n)$ do not depend on $t$.}
\end{theorem}
\begin{proof}[Remarks]

1. Let us mention that for the case of even $n$ when all eigenvalues
of $A(U)$ are complex conjugate pairs (not real), it was proved in
\cite{bialy2} that $U$ in this case must be a constant solution.
Taking into account the geometric meaning of the eigenvalues
mentioned above, one can interpret this result as a reflection of
the so called Hopf rigidity phenomenon known in Riemannian geometry.

2. We don't know if the condition of continuous factorization of
$F_p$ used in Theorem 1.1 is really essential for the result. In
smooth 1-parameter family of polynomials, roots can be chosen
continuously but in 2-parameter family this does not necessarily
hold due to monodromy effect. Continuous factorization obviously
holds in the case of all distinct roots.

\end{proof}

Our second result deals with the strictly hyperbolic case. 

In this case it appears to be a difficult problem to classify
possible smooth periodic solutions. The reason for this lies in the
fact that genuine non-linearity of the eigenvalues cannot be
guarantied in general. To emphasize this fact we introduce in this
paper a remarkable 2 by 2 reduction of the system (\ref{system}) for
$n=4$ and thus a reduction of the infinite Benney chain as well. We
prove:
\begin{theorem} The following quasi-linear system is a reduction of
(\ref{system}) for $n=4$:
$$
\begin{cases}
w_t+(wv)_x=0\\
v_t+\left(\frac{w^2}{2}-\frac{3v^2}{2}\right)_x=0.
\end{cases}
$$
It is strictly hyperbolic outside the origin in the $(u,v)$ plane.
Moreover the sign of the nonlinearity changes when $v$ is changing
sign.
\end{theorem}

We believe that understanding of the behaviour of this reduced
system can shed light in the search of smooth periodic solutions for
Hyperbolic regime in general. We prove Theorem 1.2. in Section 6 .
\section*{Acknowledgements}
It is a pleasure to thank Sasha Veselov, Misha Sodin and Inna
Scherbak for very useful discussions. It was very helpful to consult
with Jenia Shustin who also kindly communicated to us a simple proof
of Lemma 4.1 in Section 4.

\section{Riemann invariants and critical values of polynomials}
We start with the following crucial fact on critical values of smooth
family of polynomials, which we believe is of independent interest. 
\begin{theorem}
Let $F(p,y)$ be a family of polynomials smoothly depending on the
parameters $y=(y_1,...,y_n)$ of the form
$$F=\frac{1}{n+1}p^{n+1}+u_1(y)p^{n-1}+...+u_n(y).$$
Assume that $\lambda(y)$ is a (complex) critical point of the
polynomial $F$ continuously depending on $y$. Then the corresponding critical
value
$$ r(y):=F(\lambda(y),y)$$
is a $C^1$ function of $y$.
\end{theorem}
This theorem states that though in a smooth family of polynomials
the critical points are not smooth in general (due to collisions)
but the corresponding critical values are at least $C^1$.

Coming back to the formulation of Theorem 1.1, we shall denote by $\rho$ and $r_i$ the critical values of the polynomial $F$:
 $$ \rho(t,x):=F(\mu(t,x),x,t),\quad r_i(t,x):=F(\lambda_i(t,x),x,t)$$
It then
follows from Theorem 2.1 and the equation (\ref{eq:conservation})
that under the conditions of Theorem 1.1 the functions $r_i(t,x)$
and $\rho(t,x)$ are Riemann invariants for system (\ref{system}):
\begin{corollary}
Under the assumptions of Theorem 1.1 the functions $\rho$, $r_i$ are
$C^1$ and satisfy the equations
$$
  \rho_t+\mu(t,x)\rho_x=0,$$
$$ (r_i)_t+\lambda_i(t,x)(r_i)_x=0.
$$
\end{corollary}
\begin{proof}[Proof of Theorem 2.1]
One needs to show that partial derivatives of $r$ are continuous at
every point $y_0$. If $\lambda(y_0)$ is a simple root of $F_p$, that
is $F_{pp}(\lambda(y_0),y_0)\neq 0$ then the statement is obvious
because in this case it follows from the implicit function theorem
for $F_p(\lambda(y),y)=0$ that $\lambda(y)$ is continuously
differentiable in a neighborhood of $y_0$ and hence also $r(y)$ as a
superposition $F(\lambda(y),y)$. Moreover we have
$$ \partial_{y_i} r(y_0)=F_p(\lambda(y_0), y_0)+\partial_{y_i}
F(\lambda(y_0),
y_0)=\partial_{y_i}F(\lambda(y_0), y_0).
$$
We shall prove that also in general case, when $\lambda(y_0)$ is a
root of multiplicity $m>1$ the same formula holds true. We use for
this the fact that $\lambda(y)$ is H\"{o}lder continuous at a
neighborhood of $y_0$ of order $1/m$ (see \cite {polynom}). We
have for every coordinate vector $e_i$:
$$\frac{r(y_0+\varepsilon e_i)-r(y_0)}{\varepsilon}=
\frac{F(\lambda(y_0+\varepsilon e_i),y_0+\varepsilon
e_i)-F(\lambda(y_0), y_0)}{\varepsilon}=A+B,
$$
where we write:
$$
A:=\frac{F(\lambda(y_0+\varepsilon e_i),y_0+\varepsilon
e_i)-F(\lambda(y_0+\varepsilon e_i),y_0)}{\varepsilon},
$$
$$
B:=\frac{F(\lambda(y_0+\varepsilon
e_i),y_0)-F(\lambda(y_0),y_0)}{\varepsilon}.
$$
It follows from mean value theorem and the continuity of
$\lambda(y)$ that
$$
\lim_{\varepsilon\rightarrow 0} A=\partial_{y_i}F(\lambda(y_0),
y_0).
$$
Next we need to show that $B$ tends to $0$ as
$\varepsilon\rightarrow 0$. This goes as follows. Since
$\lambda(y_0)$ is a root of multiplicity $m>1$ for $F_p(p, y_0)$ it
then follows that the polynomial $F(p, y_0)$ can be written as
$$
F(p, y_0)=(p-\lambda(y_0))^{m+1}f(p)+r(y_0)
$$
for some polynomial $f$. Therefore the term $B$ can be written as:
$$
B=\frac{(\lambda(y_0+\varepsilon
e_i)-\lambda(y_0))^{m+1}f(\lambda(y_0+\varepsilon
e_i))}{\varepsilon}
.$$
Since $\lambda$ is H\"{o}lder with exponent $\frac{1}{m}$ then
$$|B|\leq\frac{ C \varepsilon^{\frac{m+1}{m}}}{\varepsilon}=
C\varepsilon^{\frac{1}{m}}.
$$
Thus $$\lim_{\varepsilon\rightarrow 0} B=0.$$ This completes the
proof.

\end{proof}
\section{Maximum principle for complex Riemann invariants}
Let $r_j(t,x)$ be a Riemann invariant corresponding to the complex
eigenvalue $\lambda_j(t,x)$. We have:

\begin{theorem} The complex Riemann invariants $r_j$ are constant
functions on the whole torus $\mathbb{T}^2$ for every $j=1,..,l$.
\end{theorem}

For simplicity we shall omit in this section the index $j$. 

Let
$r=u(t,x)+iv(t,x),\ \lambda=a(t,x) +i b(t,x), b>0$, then the
equation $(r)_t+\lambda\cdot(r)_x=0$ leads to the system of equations:
$$
\begin{cases}
  u_t+au_x-bv_x=0 \\
  v_t+bu_x+av_x=0.
 \end{cases}
$$

It is equivalent to the following:
\begin{equation}
\label{beltrami}
\begin{cases}
 u_t=\frac{a}{b}v_t+\frac{a^2+b^2}{b}v_x\\
-u_x=\frac{1}{b}v_t+\frac{a}{b}v_x.
\end{cases}
\end{equation}

This is the famous Beltrami system of equations corresponding to the
Riemannian metric
$$
ds^2=(a^2+b^2)dt^2-2adtdx+dx^2.
$$
We know that it has continuous coefficients since $\lambda(t,x)$ is
continuous by the assumptions. Moreover since $r$ is $C^1-$function
(due to Corollary 2.2) we have that $(u,v)$ is a global $C^1$
solution of the system (\ref{beltrami}). Therefore, the mapping
$(t,x)\rightarrow(u,v)$ is quasi-conformal. Then, it follows from representation theorem
for quasi-conformal maps that the maximum principle applies and
hence $(u,v)$ are constant functions on the whole torus
$\mathbb{T}^2$ (otherwise they have a point of maximum
somewhere on the torus). We refer to \cite{bers} p. 44 --
48 and \cite {ahlfors} for the details on the theory of quasi-conformal maps.
\begin{remark}
Notice that a-priori it is not true that the functions $r_j$ are
functions on the torus, but this can be achieved passing to a
suitable finite cover of the torus.
\end{remark}
\section{Using critical values as a system of local coordinates}
Consider the space $\mathbb{C}^n$ of all polynomials of degree
$(n+1)$ of the form
$$F=\frac{1}{n+1}p^{n+1}+u_1p^{n-1}+...+u_n,$$
with (complex) coefficients $u_i$. Denote by $\Lambda_{\bold{1}}\subset
\mathbb{C}^n$ the subset of polynomials with Morse critical points.
The complement to  $\Lambda_{\bold{1}}$ consists of those polynomials that the
derivative has a multiple root. We define the strata
$$\Lambda_{m_1...m_k}\subset
\mathbb{C}^n,\  \sum_{i=1}^k m_i=n,\  m_1\geq...\geq m_k\geq 1$$ to be the subset of those
polynomials $F$ such that the derivative $F_p$ has precisely $k$
distinct roots
of the multiplicities $m_i$.
The following lemma used in many papers (see for instance \cite{dubrovin}, \cite{strachan}) We could not find however a proof in the literature.
The proof below was communicated to us by Eugene Shustin.
\begin{lemma}
The critical values $r_i=F({\lambda_i}), i=1,...,k$ form a system of
local coordinates on $\Lambda_{m_1...m_k}$.
\end{lemma}
\begin{proof}
First let us remark that the claim is immediate if all the roots of
the derivative are of multiplicity 1.

In the general case proof goes as follows. First, notice the
relation
$$m_1\lambda_1+...+m_k\lambda_k=0,$$
and denote by $$D:\C^n\rightarrow\C^{n-1}$$ the linear projection
given by the derivation. In the space $\C^{n-1}$ of polynomials of
degree $n$, consider the germ $V=D(\Lambda_{m_1...m_k}$) at
$$G_0=\prod_{i=1}^k(p-\lambda_i)^{m_i}$$ of the family of
polynomials having $ k$ roots of multiplicities $m_1,..,m_k$
respectively. It follows that $V$ is smooth and its (affine-linear)
tangent space consisting  of the  polynomials $G$ which can be defined by the affine-linear equations on
$u_1,...,u_{n-1}$:
$$G^{(j)}(\lambda_i)=0, 0\leq j\leq m_i-2, i=1,...,k,$$
(here $\lambda$'s are fixed and the equation appears only when
$m_i\geq 2$). These $n-k$ linear equations are transversal in
$\C^{n-1}$, since adding $k-1$ more equations
$$G^{(m_i-1)}(\lambda_i)=0, i=1,...,k-1,$$
we obtain a system of $(n-1)$ linear equations evidently having the
unique solution $G_0$.

Hence, $\Lambda_{m_1...m_k}=D^{-1}(V)$ is a smooth germ of a
$k$-dimensional subvariety in $\C^n$, and its (affine-linear) tangent
space $W$ at each point $F_0\in D^{-1}(G_0)$ is given by $(n-k)$
(transversal) linear equations
$$W =\{F^{(j)}(\lambda_i)=0, 1\le j\le m_i-1, i=1,...,k\}.$$
Consider the map $R$ taking $F\in D^{-1}(V)$ to the set $S \in\C^k$
of critical values of $F$. The differential of $R$ at $F_0$ acts
from $W$ to $\C^k$ by formula
 $$W\ni F \mapsto (F(\lambda_1),...,F(\lambda_k)).$$
This is an affine-linear isomorphism, since the preimage of each
point is unique: the difference $H$ of two preimages of the same point
must be a polynomial of degree $\le n-1$ satisfying equations
$$H^{(j)}(\lambda_i)=0, 0\le j\le m_i-1, i=1,...,k,$$
hence $H=0$.
\end{proof}
This lemma implies the following property.
\begin{theorem}
Suppose $F(p,t)=\frac{1}{n+1}p^{n+1}+u_1(t)p^{n-1}+...+u_n(t), \
t\in I$ be a smooth curve in the space of polynomials
such that the set of critical values of $F(\cdot,t)$ does not depend
on $t$. Then $F$ is a polynomial with constant coefficients.
\end{theorem}
\begin{proof}First let us factorize the derivative $F(p,t)$:

$$
F_p(t)=\prod_i^n(p-\lambda_i(t)),
$$
where $\lambda_i(t)$ are continuous (and not necessarily distinct).
Denote by $r_i$ the corresponding critical values. Suppose that
for some open sub interval $J\subset I$ the polynomial $F(\cdot,t),\
t\in J$ lies on one stratum $\Lambda_{m_1...m_k}$. Then it follows
from the Lemma that  $F(p,t)$ is a polynomial with constant
coefficients for all $t\in J$. Indeed, this holds locally by
Lemma 4.1 and therefore on the whole $J$. So the derivative $F_t\equiv
0$ on $J$. Denote by $U$ the open subset of $I$ which is the union
of all those maximal sub-intervals for which $F(\cdot,t)$ belongs to
one stratum. As we explained above on $U$ we have $F_t\equiv 0$.
Moreover, it follows that the complement $I\backslash U$ has no
interior points and thus every point of $I\backslash U$ can be
approached by a sequence of points from $U$ and hence by continuity
$F_t$ must vanish on the whole $I$.
\end{proof}

\section{Traveling wave solutions for the system}
Now we are in position to apply the result of the previous section to the proof of Theorem 1.1.
The first step is the following
\begin{theorem}
Let $U$ be a non-constant periodic solution of (\ref{system}), then
the real eigenvalue $\mu(t,x)$ is a constant number on the whole
torus and the solution $U$ is a traveling wave solution for the
system (\ref{system}), $U=U(x-\mu t)$.
\end{theorem}
\begin{proof}
Denote by $\rho$ the Riemann invariant  corresponding to the real
eigenvalue $\mu$, i.e. $\rho(t,x)=F(\mu (t,x),x,t)$. Function $\rho$
satisfies
$$
\rho_t+\mu(t,x)\rho_x=0.
$$
Notice this equation means that $\rho$ must have constant values
along characteristic curves. These are integral curves of the
equation
$$ \dot{x}=\mu(t,x).
$$
Therefore, using Theorem 3.1, we conclude that along every
characteristic curve all Riemann invariants preserve constant
values. It then follows from Theorem 4.2 of the previous section
that all the coefficients $u_i(t,x), i=1,..,n$ are constants along
every characteristic curve. Then $\mu$ as a function of $u_i$-s is
also constant along every characteristic. But $\mu$ is the slope of
the tangent line to the characteristic, so every characteristic
curve must be a straight line. Moreover all these lines must be
parallel, since otherwise they intersect, which is impossible for
solutions of ODE. So $\mu$ is a constant. This completes the proof.

\end{proof}
The next step is to show that the only traveling wave solutions for system
(\ref{system}) are autonomous:
\begin{theorem}Let $U=U(x-\mu t)$ be a traveling wave solution of
(\ref{system}).
Then $\mu\equiv 0$, and the polynomial $F$ is a function of the
Hamiltonian $H$.
\end{theorem}
\begin{proof}
Let $U(x-\mu t)$ be a traveling wave solution for the system
(\ref{system}). Then $U^{'}(x)$ must be an eigenvector of the matrix
$A(U)$ and we have the following system of ordinary differential on
the components $u_i$:
$$
\begin{cases}
   \mu u_1^{'}\ \ = \quad \quad \quad\quad\quad \quad \quad u_2^{'}\\
   \mu u_2^{'}\ \ =-(n-1)u_1 u_1^{'}+u_3{'}\\
   \qquad \quad........\\
   \mu u_{n-1}^{'}=-2u_{n-2}u_1{'}+\ \  u_n^{'}\\
   \mu u_{n}^{'}\quad =\ -u_{n-1}u_1{'}
 \end{cases}
$$

Notice, that first $(n-1)$ equations of this system can be
integrated step by step starting from the first one and all
functions $u_2,...,u_n$ become polynomial expressions on $u_1$. We write an
additional equation using the fact that $\mu$ is an
eigenvalue of the matrix $A(U)$, so:
$$\mu^n+(n-1)u_1\mu^{n-2}+...+u_{n-1}=0.$$
Substituting polynomial expressions of $u_2,...,u_{n-1}$  into the
last equation we get the following alternative: either $\mu=0$, or
$u_1$ satisfies certain polynomial equation with constant
coefficients and therefore must be a constant. In the second case
all the components of the solution are constants. In the first case
we get that the Hamiltonian $H$ is autonomous, and $F$ turns out to be a
polynomial function of $H$. This completes the proof.
\end{proof}
\section{Hiperbolic 2 by 2 reduction for Benney chain}
In this section we derive a remarkable reduction of
system(\ref{system}) for $n=4$.

 It was shown in \cite{bialy1} that for
strictly hyperbolic case of the system (\ref{system}) the smallest
and the largest eigenvalues are genuinely non-linear. This implies
that the corresponding Riemann invariants are constants. This
motivates the following construction:

We are looking for the function
\begin{equation}
F=\frac{1}{5}p^5+u_1p^3+u_2p^2+u_3p+u_4,
\end{equation}
polynomial of degree 5
satisfying system (\ref{eq:conservation}), where $u_1=u$ is the potential. We write $F$ in the form:
\begin{equation}
F=\frac{1}{5}(p-f)^2(p-g)^2(p-a),
\end{equation}
where $f,g,a$ are some functions.
Equating the coefficients of $p^4$ and $p^3$ of the polynomials in (4) and (5)  we get
$$a=-(2f+2g),$$
$$
u=-\frac{1}{5}(3f^2+4fg+3g^2).
$$
In order to write the equations on $f,g$ notice that $p=f$ and $p=g$ are level sets of $F$
thus are invariant tori of the Hamiltonian system. Therefore the following equations hold:
$$
 \begin{cases}
  f_t+ff_x+(u)_x=0\\
  g_t+gg_x+(u)_x=0.
 \end{cases}
$$
Introduce
$$
\frac{f-g}{2}=h,\ \frac{f+g}{2}=q.
$$
We get the following equations on $h,q$
$$
  \begin{cases}
   h_t+(hq)_x=0\\
   q_t+\left(\frac{f^2}{4}+\frac{g^2}{4}-\frac{1}{5}(3f^2+4fg+3g^2)\right)_x=0.
  \end{cases}
$$
Rewriting the last equation in terms of $h$ and $q$ we come to the
system:
$$
 \begin{cases}
  h_t+(hq)_x=0\\
  q_t+\left(\frac{h^2}{10}-\frac{3}{2}q^2\right)_x=0.
 \end{cases}
$$
Finally changing $(h,q)\rightarrow(w,v),\ w=\frac{h}{\sqrt {5}}, \
v=q,$ we come to the system:
$$
\begin{cases}
 w_t+(wv)_x=0\\
 v_t+\left(\frac{w^2}{2}-\frac{3v^2}{2}\right)_x=0.
 \end{cases}
$$

 The matrix of this system reads
\[A(w,v)=\left(
\begin{array}{cc}
    v       & w \\
    w       & -3v

\end{array}
\right).
\]
The matrix is strictly hyperbolic away of the origin on the
$(w,v)$-plane. The eigenvalues are given by the following formula:
$$
\lambda_{1,2}= -v\pm R,\quad R:=\sqrt{4v^2+w^2}.
$$
In order to check  type of non-linearity one needs to check the sign
of the derivative $\frac{\partial \lambda_1}{\partial r_1}.$ The
eigenvector of $\lambda_1$ is given by $\xi_1=(w, R-2v)$. Therefore we
compute:
$$\frac{\partial \lambda_1}{\partial r_1}=d\lambda_1(\xi_1)=\partial_w(\lambda_1)w+\partial_v(\lambda_1)(R-2v)=$$
$$\quad = \frac{w^2}{R}+\left(-1+\frac{4v}{R}\right)(R-2v)=
\frac{6v(R-2v)}{R}.
$$
So we have the sign of the derivative $\frac{\partial
\lambda_1}{\partial r_1}$ equals that of $v$. This proves Theorem
1.2.

\end{document}